\documentclass[11pt,a4paper]{article}

		\usepackage{a4}   

		\usepackage{fancyhdr}
	
		\usepackage[english,ngerman]{babel} 
		
			\usepackage[fleqn]{amsmath}
		\usepackage{amsthm}
		\usepackage{amssymb}
		\usepackage{amstext}
		\usepackage{amsfonts}
		\usepackage{amsbsy}
		\usepackage{amscd}
		\usepackage{latexsym}
		\usepackage{mathtools}

		\usepackage{makeidx}
		\makeindex 
	
		\usepackage[utf8]{inputenc}
		
		\usepackage[T1]{fontenc}

\usepackage{dsfont}

		\usepackage{lmodern}
	
		\usepackage{listings}

		\usepackage{booktabs,longtable,tabularx}
		\usepackage{url}
		\usepackage{enumerate}
		\newtheoremstyle{normalstyle}
			{.5em}
			{3pt}
			{\itshape}
			{}
			{\normalfont\bfseries}
			{.\newline}
			{.5em}
			{}
		\theoremstyle{normalstyle}
		\newtheorem{thm}{Theorem}
		\newtheorem{lem}[thm]{Lemma}
		\newtheorem{definition}[thm]{Definition}

\newcommand{\Z}{\mathbb{Z}}

\newcommand{\R}{\mathbb{R}}

\newcommand{\id}{\mathds{1}}

\newcommand{\norm}[1]{\left|\left|#1\right|\right|}
\newcommand{\skp}[1]{\left\langle #1 \right\rangle}

\newcommand{\df}{\mathrel{\mathop:}=}

\title{Deciding whether a Lattice has an Orthonormal Basis is in co-NP}

\author{Christoph Hunkenschr{\"o}der\\
	    {\'E}cole Polytechnique F{\'e}d{\'e}rale de Lausanne\\
	    christoph.hunkenschroder@epfl.ch}

\makeindex

\usepackage{anysize}

\usepackage{forest}
\usepackage{relsize}
\usetikzlibrary{trees}
\usetikzlibrary{matrix}

\begin{document}
\selectlanguage{english}

	\maketitle	


\begin{abstract}
We show that the problem of deciding whether a given Euclidean lattice $L$ has an orthonormal basis is in NP and co-NP.
Since this is equivalent to saying that $L$ is isomorphic to the standard integer lattice, this problem is a special form of the Lattice Isomorphism Problem, which is known to be in the complexity class SZK.
\end{abstract}

\section{Introduction}
Let $B \in \R^{n \times n}$ be a non-singular matrix generating a full-dimensional Euclidean lattice $\Lambda(B) = \{B z \mid z \in \Z^n\}$.
Several problems in the algorithmic theory of lattices, such as the covering radius problem, or the shortest or closest vector problem,  become easy if the columns of $B$ form an orthonormal basis.
However, if $B$ is any basis and we want to decide whether there is an orthonormal basis of $\Lambda(B)$, not much is known.
As this is equivalent to $\Lambda(B)$ being a rotation of $\Z^n$, we call this decision problem the \emph{Rotated Standard Lattice Problem} (RSLP), which is the main concern of the work at hand.
The goal of this article is twofold.
For one, we want to show that the RSLP is in NP $\cap$ co-NP.
We will use a result of Elkies on \emph{characteristic vectors}~\cite{elkies1995characterization}, which appear in analytic number theory.
It seems that characteristic vectors are rather unknown in the algorithmic lattice theory, mayhap because it is unclear how they can be used.
The second attempt of this paper is thus to introduce Elkies' result and characteristic vectors to a wider audience.

\subsubsection*{Related work}
In~\cite{lenstra2017symmetry}, Lenstra and Silverberg show that RSLP can be decided in polynomial time, provided that additional information on the automorphism group of the lattice is part of the input.
However, they do not discuss the complexity of this problem in general, though their results might allow for a co-NP certificate for the general case as well.
When the input lattice is a construction-A lattice, Chandrasekaran, Gandikota and Grigorescu show that existence of an \emph{orthogonal} basis can be decided in polynomial time~\cite{kartik2016constr-A}.
If it exists, they also find one.

The RSLP can be seen as a special case of the Lattice Isomorphism Problem,
which, given two lattices $\Lambda_1,\Lambda_2$, asks whether there is an isomorphism
$\varphi: \Lambda_1 \rightarrow \Lambda_2$ between the two lattices that preserves the Euclidean structure ($\skp{x,y} = \skp{ \varphi(x), \varphi (y)}$).
That is, $\skp{x,y} = \sum_{i=1}^n x_i y_i$ for $x,y \in \R^n$.

The Lattice Isomorphism Problem was first introduced by Plesken and Souvignier~\cite{plesken1997computing}, solving it in small dimension for specific lattices of interest.
In~\cite{dutour2009complexity}, Dutour Sikiri{\'c}, Sch{\"u}rmann and Vallentin show that this problem is at least as hard as the more famous Graph Isomorphism Problem.
The best algorithm for the Lattice Isomorphism Problem the author is aware of is due to Haviv and Regev, and has a running time of $n^{\mathcal{O} (n)}$~\cite{haviv2014lattice}.
They solve the problem by computing all orthogonal linear transformations between the two given lattices $\Lambda_1$, $\Lambda_2$.
On the complexity side, they show that the problem is in the complexity class SZK (Statistical Zero Knowledge), which already suggests that it is not NP-hard.


Sufficient background for the topic is provided in the next section, where we also introduce characteristic vectors and show some easy properties.
Then, we show that RSLP is in NP $\cap$ co-NP, and conclude with open questions.

\section{Preliminaries}
In the following, we will provide sufficient background on lattices.
For more details and proofs, we refer to the textbook of Gruber \& Lekkerkerker~\cite{gruber1993geometry}.

For linearly independent vectors $b_1,\dots,b_n \in \R^n$,
the (full-dimensional) \emph{lattice} $\Lambda \subseteq \R^n$ \emph{generated by} $b_1,\dots,b_n$ is the set
\[
\Lambda = \left\{\sum_{i=1}^n \alpha_i b_i \mid \forall i \in \{1,\dots,n\} : \, \alpha_i \in \Z \right\}.
\]
The matrix $B = (b_1,\dots,b_n)$ is called a \emph{basis} of $\Lambda$.
It is known that two bases $B_1,B_2$ generate the same lattice, if and only if there exists a unimodular matrix $U \in \Z^{n \times n}$ such that $B_1 = B_2 U$.
Two lattices $\Lambda_1, \Lambda_2$ are \emph{isomorphic}, in symbols $\Lambda_1 \cong \Lambda_2$, if there exists an orthogonal matrix $Q \in \R^{n \times n}$ such that $\Lambda_1 = Q \Lambda_2$.
In this case, it follows that $B_2$ is a basis of $\Lambda_2$ if and only if $B_1 \df QB_2$ is a basis of $\Lambda_1$.
%

We are interested in the following problem, which is a special case of the Lattice Isomorphism Problem.

\vspace*{5pt} 
\noindent\fbox{
\parbox{0.97\textwidth}{
\textbf{Rotated Standard Lattice Problem}\\
\begin{tabular}{rl}
\textsc{Instance:} & A lattice $\Lambda \subseteq \R^n$, given by a basis $B \in \R^{n \times n}$. \\
\textsc{Task:} & Decide whether $\Lambda \cong \Z^n$.
\end{tabular}
}}
\vspace*{5pt}

The attentive reader might notice that an isomorphism refers to an orthogonal matrix $Q$, whereas a rotation usually refers to an orthogonal matrix with positive determinant, $\det (Q) = 1$.
However, in our setting, i.e.\ one of the lattices in consideration is fixed to $\Z^n$,
the terms turn out to be equivalent for the following reason.
A matrix $B$ generates $\Z^n$, if and only if $B$ is unimodular.
Given an isomorphic lattice with basis $QB$, we can multiply the first column of $Q$ and the first row of $B$ by $-1$, without changing the basis $QB$.
This flips the sign of $\det(Q)$, while $B$ remains unimodular.
Hence, $\Lambda(QB)$ is indeed a rotation of $\Z^n$.

Though a lattice is usually specified by a basis matrix $B$, we will see that another representation is preferable for this problem.

The \emph{Gram matrix} $G$ of a basis $B$ is defined as $G = B^\intercal B$, i.e.\ $G_{i,j} = \skp{b_i,b_j}$.
An advantage of the Gram matrix is that it ``forgets'' the embedding of a lattice $\Lambda$ into the Euclidean space, and only carries the information of the isomorphism class of $\Lambda$, as implied by the following lemma.

\begin{lem}
\label{lem:gram-iso}
Two bases $B_1,B_2 \in \R^{n \times n}$ generate isomorphic lattices $\Lambda(B_1),\Lambda(B_2) \subseteq \R^n$, if and only if there exists a unimodular matrix $U \in \Z^{n \times n}$ such that for the corresponding Gram matrices $G_1,G_2$ the relation
$G_1 = U^\intercal G_2 U$ holds.
\end{lem}
\begin{proof}
If there is an isomorphism $\Lambda(B_1) = Q \Lambda(B_2)$, then $QB_2$ is a basis of $\Lambda_1$.
This implies that there is a unimodular matrix $U \in \Z^{n \times n}$ such that $B_1 = QB_2 U$, and we obtain $G_1 = U^\intercal G_2 U$.

On the other hand, if $G_1 = U^\intercal G_2 U$, define $Q = B_1 U^{-1} B_2^{-1}$, and verify
\[
Q^\intercal Q = B_2^{-\intercal} ( U^{-\intercal}  B_1^\intercal B_1 U^{-1}) B_2^{-1} = B_2^{-\intercal} G_2 B_2^{-1} = \id.
\]
Hence, we find $Q \Lambda(B_2) = \Lambda(QB_2) = \Lambda(B_1 U^{-1}) = \Lambda(B_1)$, since $U^{-1}$ is again unimodular.
\end{proof}
Clearly, a Gram matrix is always symmetric and positive definite.
Since every symmetric and positive definite matrix has a Cholesky decomposition, we can consider every symmetric and positive definite matrix as a Gram matrix of some lattice basis.
Moreover, the lattice $\Z^n$ is generated by the identity, whose Gram matrix is in addition unimodular.
This grants a reduction to the following problem.

\vspace*{5pt} 
\noindent\fbox{
\parbox{0.97\textwidth}{
\textbf{Unimodular Decomposition Problem (UDP)}\\
\begin{tabular}{rl}
\textsc{Instance:} & A symmetric, positive definite, unimodular matrix $G \in \Z^{n \times n}$. \\
\textsc{Task:} & Decide whether there exists a unimodular matrix $U \in \Z^{n \times n}$ with $G = U^\intercal U$.
\end{tabular}
}}
\vspace*{5pt}

Clearly, UDP is in NP.
A certificate is simply given by the matrix $U$,
whose entries are bounded by $\max \{ \sqrt{G_{ii}} \mid 1 \leq i \leq n\}$.
In the following, we will discuss a paper of Elkies~\cite{elkies1995characterization} and apply his results to show that the problem is also in co-NP.

To this end, let us recall some terminology for a full-dimensional lattice $\Lambda \subseteq \R^n$.
The dual of $\Lambda$, denoted by $\Lambda^\star$, is defined as
$\Lambda^\star = \{y \in \R^n \mid \forall \, x \in \Lambda: \, y^\intercal x \in \Z \}$.
We call $\Lambda$ \emph{self-dual} if $\Lambda = \Lambda^\star$.
Self-dual lattices are also called unimodular lattices for the following reason.
\begin{lem}
A matrix $B \in \R^{n \times n}$ generates a self-dual lattice $\Lambda(B)$ if and only if the corresponding Gram matrix $G = B^\intercal B$ is unimodular.
\end{lem}
\begin{proof}
Let $B$ be the basis of a self-dual lattice.
Then $B^{-\intercal }$ is a basis as well, and there exists a unimodular matrix $G$ such that $B^{-\intercal} G = B$, which is equivalent to $G = B^\intercal B$.

Let $G = B^\intercal B$ be unimodular, and $x = Bz_1$ and $y= Bz_2$ be any two lattice vectors.
Since $x^\intercal y = z_1^\intercal G z_2$, we have $\Lambda(B) \subseteq \Lambda(B)^\star$.
Since $\det (\Lambda(B))^2 = \det (B)^2 = \det (G) = 1$, they have to be the same already.
\end{proof}
This also implies that in a self-dual lattice, the scalar product of any two vectors is an integer.

\begin{definition}
A vector $w \in \Lambda$ of a self-dual lattice $\Lambda$ is said to be characteristic, if
\[
\forall \, v \in \Lambda: \quad \skp{v,w} \equiv \skp{v,v} \!\mod 2.
\]
For $x \equiv y \!\mod k$, we will also write $x \equiv_k y$ for short.
\end{definition}

It is known that for dimensions $n \leq 7$, the lattice $\Z^n$ is the unique self-dual lattice (up to isomorphism)~\cite{conwaysloane1999splag}.
In dimension $8$, the lattice
\[
E_8 = \left\lbrace z \in \R^8 \left\vert \sum_{i=1}^8 z_i \equiv_2 0, z \in \Z^8 \cup \left( \frac12 \mathbf{1} + \Z^8 \right) \right. \right\rbrace
\]
is self-dual, but not isomorphic to $\Z^8$.
Here $\mathbf{1} = (1,\dots,1)^\intercal$ denotes the all-one vector.

Before we turn our attention to the main result, let us show some basic properties of characteristic vectors.

\newpage
\begin{lem}
The following are true for every self-dual lattice $\Lambda = \Lambda^\star \subseteq \R^n$.
\begin{enumerate}[\indent $i)$]
\item There exists a characteristic vector $w \in \Lambda$.
\item The set of characteristic vectors is precisely a co-set $w + 2\Lambda$, where $w \in \Lambda$ is any characteristic vector.
\item For any two characteristic vectors $u,w \in \Lambda$, we have $\norm{u}^2 - \norm{w}^2 \equiv_8 0$.
\item If we are given a Gram matrix $G$ of $\Lambda$, we can compute a vector $z \in \Z^n$ such that for all $y \in \Z^n$, we have $y^\intercal G z \equiv_2 y^\intercal G y$ in polynomial time.
\item The shortest characteristic vectors of the lattice $\Z^n$ are the vectors $\{-1,1\}^n$.
\end{enumerate}
\end{lem}
\begin{proof}
\leavevmode
\begin{enumerate}[\indent $i)$]
\item Let $B = (b_1,\dots,b_n)$ be a basis of $\Lambda$, and $D =  (d_1,\dots,d_n) = B^{- \intercal}$ the corresponding dual basis, also spanning $\Lambda$.
Represented in the primal basis, define a vector $w = \sum_{i=1}^n \norm{d_i}^2 b_i \in \Lambda$, and let $v = \sum_{i=1}^n \alpha_i d_i \in \Lambda$ be any lattice vector, represented in the dual basis $D$.
Using $x^2 \equiv_2 x$ for $x \in \Z$, we obtain
\begin{align*}
\skp{v,v}
	&= \sum_{i=1}^n \sum_{j=1}^n \alpha_i \alpha_j d_i^\intercal d_j
	\equiv_2 \sum_{i=1}^n \alpha_i^2 \norm{d_i}^2 \\
	&\equiv_2 \sum_{i=1}^n \alpha_i \norm{d_i}^2
	= \sum_{i=1}^n \sum_{j=1}^n \alpha_i \norm{d_j}^2 d_i^\intercal b_j
	= \skp{v,w}.
\end{align*}
Thus, $w$ is a characteristic vector and $i)$ is shown.
\item Now let $w$ be as in the first part, and $w^\prime$ be any characteristic vector.
Since for any $y \in \Lambda$, we have $\skp{ w^\prime + 2y, v}= \skp{w^\prime,v} + 2 \skp{y,v} \equiv_2 \skp{w^\prime,v}$, the whole co-set $w^\prime + 2\Lambda$ consists of characteristic vectors.

If $w^\prime$ has the representation $w^\prime = \sum_{i=1}^n \gamma_i b_i \in \Lambda$, computing
\[
\norm{d_k}^2 = \skp{d_k,d_k} \equiv_2 \skp{w^\prime,d_k} = \gamma_k
\]
for every coefficient $\gamma_k$ shows that $w^\prime \in w + 2\Lambda$, finishing the proof of point $ii)$.
\item Let $w = Bc$ be a characteristic vector.
It suffices to show that we have $\norm{w}^2 - \norm{w + 2b_k}^2 \equiv_8 0$ for $k=1,\dots,n$.
The claim then follows by repeatedly adding or subtracting twice a basis vector.
Let $G = B^\intercal B$, and compute
\[
\norm{w + 2 b_k}^2
	= (c + 2 e_k)^\intercal G (c + 2e_k)
	= c^\intercal G c + 4 \underbrace{c^\intercal G e_k}_{\equiv_2 e_k^\intercal G e_k} + 4 e_k^\intercal G e_k
	\equiv_8 \norm{w}^2 + 8 G_{kk},
\]
where the emphasized equivalence follows from $w$ being a characteristic vector.
Since $G_{kk} \in \Z$, we are done.
\item Observe that $G^{-1} = D^\intercal D$;
hence, by the definition of $w^\prime$, we can compute $G^{-1}$ and set $z_k = (G^{-1})_{kk}$ for $k = 1,\dots,n$.
\item Choosing the identity matrix as lattice basis, it follows that all characteristic vectors of $\Z^n$ must have odd entries only, hence point $v)$ follows.
\end{enumerate}
\end{proof}

The computations carried out in part $ii)$ are already discussed in~\cite{gerstein2004characteristic}.
Point $iii)$ can also be found in~\cite[Chap. V]{serre2012course}.
Both sources are written in terms of quadratic forms.

\section{The Unimodular Decomposition Problem is in NP and co-NP}
We are now able to provide the main result of this article.
The crucial argument we are using will be 
Elkies' theorem, which reads as follows.

\begin{thm}[Elkies~\cite{elkies1995characterization}]
Let $\Lambda$ be a unimodular lattice in $\R^n$ with no characteristic vector such that $\Vert w \Vert^2 < n$.
Then $\Lambda \cong \Z^n$.
\end{thm}


Our result can now be proven as follows.
Assume we are given a Gram matrix $G$ such that there is no unimodular matrix $U$ with $G = U^\intercal U$.
We can interpret $G$ as the Gram matrix of some lattice, for instance by considering the Cholesky decomposition $G = L^\intercal L$.
Since $G \neq U^\intercal U$, $\Lambda(L) \not\cong \Z^n$, and by Elkies, there is a characteristic vector $w \in \Lambda(L)$ with $\Vert w \Vert^2 < n$.
Our certificate will be the coefficient vector of $w$, i.e.\ the vector $z$ such that $w = Lz$.
In Lemma~\ref{lem:basis-check}, we will show that $z$ is independent of the matrix $L$, and it suffices to check the parity condition of a characteristic vector on the $n$ vectors of an arbitrary basis $B$, instead of all lattice vectors.
In Lemma~\ref{lem:bit-complexity}, we will show that the coefficients of $w$ appearing in $z$ cannot be too large, i.e.\ the bit complexity of the certificate is polynomially bounded by the input size.
\begin{lem}
\label{lem:basis-check}
Let $G \in \Z^{n \times n}$ be a symmetric, positive definite, and unimodular matrix.
We have 
\[
e_k^\intercal G e_k \equiv_2 e_k^\intercal G z, \quad \forall \, k=1,\dots,n,
\]
if and only if for every matrix $B$ with $B^\intercal B =  G$, the vector $w= Bz$ is a characteristic vector in the lattice $\Lambda(B)$.
\end{lem}
\begin{proof}
Let $x = \sum_{i=1}^n \alpha_i b_i \in \Lambda$ be any vector, and $w = Bz \in \Lambda$.

If $w$ is a characteristic vector, we have 
$e_k^\intercal G e_k = \skp{b_k,b_k} \equiv_2 \skp{b_k, w} = e_k^\intercal z$.

For the other direction, we find
\begin{align*}
\skp{x, w - x} &= \skp{ \sum_{i=1}^n \alpha_i b_i, w - \sum_{i=1}^n \alpha_i b_i} = \skp{\sum_{i=1}^n \alpha_i b_i, v} - \sum_{i} \sum_{j} \alpha_i \alpha_j \skp{b_i,b_j} \\
	&\equiv_2 \sum_{i} \alpha_i (e_i^\intercal G z)) - \sum_i \alpha_i^2 (e_i^\intercal G e_i) \equiv_2 \sum_i (\alpha_i - \alpha_i^2) e_i^\intercal G e_i \\
	&\equiv_2 0,
\end{align*}
showing that $w$ is indeed a characteristic vector.
\end{proof}
\begin{lem}
\label{lem:bit-complexity}
Let $G \in \Z^{n \times n}$ be a symmetric, positive definite, and unimodular matrix, and $z \in \Z^n$ such that $z^\intercal G z \leq n$ and  
\[
e_k^\intercal G e_k \equiv_2 e_k^\intercal G z, \quad \forall \, k=1,\dots,n.
\]
Then the bit complexity of $z$ is polynomially bounded by the bit complexity of $G$.
\end{lem}
\begin{proof}
Let $M \in \Z$ be a bound on the entries of $G$, i.e.\ $| G_{i,j} | \leq M$ for all $i,j$ in range, and let $v_1,\dots,v_n$ be an orthonormal basis of eigenvectors with eigenvalues $0 \leq \lambda_1 \leq \dots \leq \lambda_n \leq nM$.
Since $\prod_{i=1}^n \lambda_i = \det (G) = 1$, the inequality $\lambda_1 \geq 1/(nM)^{n-1}$ holds.
Writing $z = \sum_{i=1}^n \alpha_i v_i$, we estimate
$
n \geq z^\intercal G z = \sum_{i=1}^n \alpha_i^2 \lambda_i,
$
implying $\alpha_i^2 \leq \frac{n}{\lambda_i} \leq (nM)^n$.
Thus, $\norm{z}^2 \leq (nM)^{n+1}$, and since $z \in \Z^n$, its bit-complexity is bounded by $\mathcal{O} (n (\log (n) + \log(M) ))$.
\end{proof}
\begin{thm}
\label{thm:co-np}
The Unimodular Decomposition Problem (and thus the Rotated Standard Lattice Problem) is in NP $\cap$ co-NP.
\end{thm}
\begin{proof}
If $G$ is a yes-instance, the certificate is the unimodular matrix $U$ such that $U^\intercal U = G$.
Since $G_{i,i} = U_i^\intercal U_i$, where $U_i$ is the $i$-th column, all entries in $U$ are bounded by $\max \{\sqrt{ G_{ii}} \mid i=1,\dots,n\}$, hence the encoding size is polynomial in the input, and verifying the certificate takes polynomial time.

If $G$ is a no-instance, 
by Elkies' Theorem and Lemma~\ref{lem:basis-check}, there exists a vector $z \in \Z^n$ corresponding to a short characteristic vector.
By Lemmas~\ref{lem:basis-check} and~\ref{lem:bit-complexity}, we can verify that $z$ is a certificate in polynomial time.
\end{proof}

\section{Conclusion}
We have seen that characteristic vectors are well suited as a certificate for the RSLP, and their coefficient vectors for the UDP.
If $\Lambda$ is given by its basis, it follows from the discussions that the problem at hand can be solved by computing the co-set of characteristic vectors, together with a single call to an oracle for the Closest Vector Problem (CVP), computing the shortest among all characteristic vectors.

However, CVP is NP-hard, and the best known running time using this reduction we are aware of is $2^{\mathcal{O}(n)}$~(see e.g.~\cite{aggarwal2018just}).

Another easy approach to the RSLP is the following.
If a lattice admits an orthogonal basis $\{b_1,\dots,b_n\}$,
then any basis reduced in the sense of Hermite, Korkine and Zolotareff~\cite{korkine1873sur, hermite1850extraits} is an orthogonal basis.
Due to the recursive structure of those bases, $n$ calls to an oracle for the Shortest Vector Problem (SVP) are sufficient to find this basis.
Hence the best known running time using this reduction is $2^{\mathcal{O} (n)}$~(see e.g.~\cite{aggarwal2018just}), but allows to find general \emph{orthogonal} bases.

However, as this problem is -- from a complexity point of view -- supposedly easier than SVP or CVP, it is of great interest to see a smarter algorithm.
It seems to be a reappearing phenomenon that if a class $\mathcal{C}$ of lattices allows to solve lattice problems such as $SVP$ or $CVP$ fast, a certain representation is needed.
For instance, if $\Lambda$ is a lattice of Voronoi's first kind, we still need to know an obtuse superbasis before using the polynomial time algorithm in~\cite{mckilliam2014finding}.
The only exception the author is aware of are construction-A lattices~\cite{kartik2016constr-A}.
Therefore, being able to detect orthonormal, or even orthogonal bases is an interesting open problem.


\subsection*{Acknowledgements}
This work was supported by the Swiss National Science Foundation (SNSF) within the project \emph{Lattice Algorithms and Integer Programming (Nr. 185030)}.

The author thanks Georg Loho and Matthias Schymura for helpful discussions, and his advisor Fritz Eisenbrand, who brought the problem to the author's attention.




		\bibliographystyle{acm}
		\bibliography{bibliography}


\end{document}